\newif\ifarxiv
\arxivtrue 

\newif\iftodos
\todostrue  

\newif\ifcomment
\commenttrue      

\newif\ifanonomous

\ifarxiv

\documentclass{article}

\usepackage[a4paper, total={6.5in, 10in}]{geometry}
\usepackage[T1]{fontenc}
\usepackage{babel}

\usepackage{authblk}

\usepackage{hyperref}
\hypersetup{
    colorlinks=true,
    linkcolor=blue,
    filecolor=magenta, 
    urlcolor=cyan,
    pdftitle={Towards Quantum Universal Hypothesis Testing},
    pdfpagemode=FullScreen,
}

\else

\newcommand{\CLASSINPUTtoptextmargin}{1in}
\newcommand{\CLASSINPUTbottomtextmargin}{1.02in}
\documentclass[conference]{IEEEtran}
\IEEEoverridecommandlockouts

\usepackage{url}

\fi

\usepackage{cite}

\usepackage{amsmath,amssymb,amsfonts,amsthm}

\usepackage{algorithmic}
\usepackage{graphicx}
\usepackage{textcomp}
\usepackage{xcolor}

\usepackage{subfloat}

\usepackage{mathtools}
\usepackage{bbm}

\usepackage[nameinlink]{cleveref}

\usepackage[utf8]{inputenc}
\usepackage[T2A,T1]{fontenc}
\usepackage{aricksMacros}

\usepackage[normalem]{ulem}
\useunder{\uline}{\ul}{}

\def\BibTeX{{\rm B\kern-.05em{\sc i\kern-.025em b}\kern-.08em
    T\kern-.1667em\lower.7ex\hbox{E}\kern-.125emX}}

\definecolor{GV}{RGB}{255, 105, 180}
\definecolor{Biao}{RGB}{251, 139, 35}
\definecolor{Jason}{RGB}{192, 194, 201}

\newcommand{\sh}{\hat{\sigma}}
\newcommand{\rh}{\hat{\rho}}

\newcommand{\sx}{\sigma_X}
\newcommand{\sy}{\sigma_Y}
\newcommand{\sz}{\sigma_Z}
\newcommand{\sk}{\sigma_K}
\newcommand{\shx}{\hat{\sigma}_X}
\newcommand{\shy}{\hat{\sigma}_Y}
\newcommand{\shz}{\hat{\sigma}_Z}
\newcommand{\shk}{\hat{\sigma}_K}
    
\begin{document}

\title{Towards Quantum Universal Hypothesis Testing

}

\ifarxiv

\author[1]{Grootveld, Arick}

\author[1]{Yang, Haodong}

\author[1]{Chen, Biao
\thanks{B. Chen would like to acknowledge the support of a Google gift fund}}

\author[1]{Gandikota, Venkata}

\author[1]{Pollack, Jason}

\affil[1]{Syracuse University, College of Engineering and Computer Science}


\fi

\maketitle

\begin{abstract}

Hoeffding's formulation and solution to the universal hypothesis testing (UHT) problem had a profound impact on many subsequent works dealing with asymmetric hypotheses. In this work, we introduce a quantum universal hypothesis testing framework that serves as a quantum analog to Hoeffding's UHT. Motivated by Hoeffding's approach, which estimates the empirical distribution and uses it to construct the test statistic, we employ quantum state tomography to reconstruct the unknown state prior to forming the test statistic. Leveraging the concentration properties of quantum state tomography, we establish the exponential consistency of the proposed test: the type II error probability decays exponentially quickly, with the exponent determined by the trace distance between the true state and the nominal state.

\end{abstract}

\section{Introduction}
\label{section:intro}

Quantum hypothesis testing (QHT) dates back to 1967, when Helstrom formally introduced quantum detection theory for ``choosing one of two density operators as the better description of the state" \cite{HELSTROM1967254}. The projection detector proposed by Helstrom reduces to the classical likelihood ratio test when the two density matrices commute. Since then,  many significant results have been discovered in QHT. Examples include Helstrom's bound on minimum error probability 
in distinguishing two quantum states~\cite{Helstrom1969} and subsequent development in obtaining asymptotic performence bounds under different setups (see, e.g., \cite{audenaert2007discriminating}). 
Perhaps the most celebrated result is the  Quantum Stein's lemma~\cite{hiai1991proper}, which - thanks to the strong converse established in~\cite{ogawa2000strong} - provides a clean result that serves as a precise quantum analog of the classical Stein's lemma \cite{EInfo06}.

This parallel - both in problem formulations and in solutions - between classical and QHT is quite apparent in many of the results in the literature. 
Indeed, this observation is what motivates the present work: although there is a line of work focused on it in classical hypothesis testing, the so-called universal hypothesis testing (UHT) has remained largely unexplored in the quantum setting. UHT addresses the problem of determining whether a sequence of samples is generated from a known nominal distribution. It is, in essence, equivalent to the goodness of fit testing~\cite{CIT-004, EInfo06, Romano2012, Pearson01071900}. Hoeffding's formulation for UHT~\cite{hoeffding1965asymptotically}, in which the unknown alternative is a fixed distribution in the probability simplex, leads to a surprising result: even though this alternative is completely unknown, the achievable performance (in terms of the type II error probability, or more precisely, its exponential decay rate) is essentially the same as that of simple hypothesis testing when the alternative is fully specified.

The generality of the UHT framework makes it particularly appealing for a broad class of problems encountered in both classical and quantum information systems in which the goal is to verify or authenticate the distribution/state when given multiple copies of samples/quantum systems. The present work defines the quantum universal hypothesis testing (QUHT) problem and provides a simple solution to this problem through the use of quantum state tomography (QST). Our approach is largely motivated by Hoeffding's treatment of the classical UHT, which is, using modern terminology, a generalized likelihood ratio test (GLRT). With a GLRT, the empirical distribution is estimated using the sample sequence and used in constructing the likelihood ratio test that simplifies to the Kullback-Leibler divergence (KLD) between the empirical distribution and the nominal distribution. The concentration of the empirical distribution around the true distribution, as the number of samples increases, leads to the desired error probability performance. 

While QUHT has not been formally defined or studied to date, there have been several related works reported in the quantum literature. These include extensions of Sanov’s theorem~\cite{bjelakovic2005quantum,audenaert2008asymptotic,notzel2014hypothesis,hayashi2024another}, exentions of classical $t$ and $F$ tests for Gaussian quantum states ~\cite{kumagai2011quantum}, composite hypothesis testing~\cite{brandao2020adversarial,berta2021composite, fujiki2025quantum}, and Hoeffding-type error exponents in binary quantum testing~\cite{hayashi2007error,nagaoka2006converse,audenaert2008asymptotic}. Additionally, there are several lines of research in quantum circuits and information where QUHT can play an important role. Examples include quantum circuit equivalence~\cite{long2024equivalence,burgholzer2020advanced}, quantum state discrimination~\cite{barnett2009quantum,bae2015quantum,chefles2000quantum,bergou2010discrimination}, and quantum state verification and certification~\cite{huang2024certifying,yu2022statistical,buadescu2019quantum}.

The UHT formulation described above  is often referred to as the one-sample test, especially in the machine learning (ML) community. One can similarly define the quantum analog to the two-sample test~\cite{gretton2012kernel} in which two sequences of samples (quantum systems) are given and the task is to determine whether the two sequences are generated from the same distribution (quantum state). Surprisingly, even in the classical setting the optimal asymptotic error exponent of two sample testing was only recently shown, specifically in the continuous distribution setting  \cite{zhu2019universal} \cite{zhu2021asymptotically}. 
Our proposed solution to the QUHT can be easily generalized to the two-sample setting, since respective concentrations of the two-sample sequences will ensure exponential consistency of the hypothesis test when QST is applied to both sequences.

\section{Problem Setup}
\label{section:ProblemSetup}

Let \(\mathcal{H}\) be a finite‑dimensional Hilbert space. We denote by \(\mathcal{L}(\mathcal{H})\) the space of linear operators on \(\mathcal{H}\) and by \(\mathcal{U}(\mathcal{H})\subset\mathcal{L}(\mathcal{H})\) the unitary operators. The set of density operators \(\mathcal{D}(\mathcal{H})\subset\mathcal{L}(\mathcal{H})\) consists of positive semidefinite, unit‑trace operators~\cite{wilde2013quantum}. In the special case \(\dim(\mathcal{H}_2)=2\), \(\mathcal{D}(\mathcal{H}_2)\) is the set of qubit density operators. We write \(\Tr(\cdot)\) for the trace functional, and for any \(\sigma\in\mathcal{D}(\mathcal{H})\), its \(m\)-fold tensor product is denoted \(\sigma^{\otimes m}\in\mathcal{D}(\mathcal{H}^{\otimes m})\). Throughout, \(\ln\) denotes the logarithm base \(e\). The \textbf{trace norm} of \(\rho\in\mathcal{L}(\mathcal{H})\) is \(\|\rho\|_1=\Tr[\sqrt{\rho^\dagger\rho}]\), and the \textbf{trace distance} between \(\rho\) and \(\sigma\) is \(T(\rho,\sigma)=\tfrac12\|\rho-\sigma\|_1\). The \textbf{quantum fidelity} is \(F(\rho,\sigma)=\|\sqrt{\rho}\,\sqrt{\sigma}\|_1^2\), and the \textbf{quantum relative entropy} is \(D(\rho\Vert\sigma)=\Tr[\rho\ln\rho-\rho\ln\sigma]\)~\cite{wilde2013quantum}. We use $D_{KL}(P\|Q)$ to denote the classical Kullback-Leibler divergence.

\subsection{Quantum Hypothesis Testing}
\label{subsection:QHT}

Analogous to the classical setting, in QHT, independent copies of quantum systems are generated from the same quantum state $\sigma$, and the task is to discern if $\sigma = \rho_0$ or  $\sigma = \rho_1$,
where $\rho_0\neq  \rho_1 \subset \mathcal{D}(\mathcal{H})$. 
Or formally,
\begin{align}
    \label{equation:BinaryQuantumHypothesisTesting}
    \begin{split}
    H_0: \sigma &= \rho_0\\
    H_1: \sigma &=\rho_1.
    \end{split}
\end{align}

Any test is specified by a decision rule $M_m: \mathcal{D}(\mathcal{H})^m \to \{0, 1\}$, such that 
\begin{equation}
    \hat{H} = \begin{cases}
        H_0, & \text{if }\  M_m(\sigma^{\otimes m}) = 0\\
        H_1, & \text{if }\  M_m(\sigma^{\otimes m}) = 1
    \end{cases}
\end{equation}
Here we use $\hat{H}$ to denote the declared hypothesis and will reserve $H$ for the true hypothesis. 

For the decision rule described by $M_m$ define the type I error to be $\alpha(M_m) := \Prob_{M_m}[\hat{H} = H_1|H= H_0]$ and the type II error to be $\beta(M_m) := \Prob_{M_m}[\hat{H} = H_0 | H = H_1]$. 
The goal of any good testing strategy would be to minimize these errors, with respect to some constraints. Two common paradigms for developing tests are the so called ``symmetric'' and ``asymmetric'' \cite{cheng2024invitation} hypothesis testing settings. 
In the symmetric setting the goal is to simply minimize the total error or, when prior probabilities are available, the average error probability, whereas in the asymmetric setting, the goal is to minimize the type II error subject to a constraint on the Type I error. In this work we will be primarily concerned with the asymmetric setting, with minor attention paid to the symmetric setting. In classical literature the symmetric setting is sometimes referred to as the Bayesian setting, and the asymmetric setting is classically known as the Neyman-Pearson setting.

With a single (or finite) copy of the quantum systems, the Helstrom bound \cite{holevo1973statistical} \cite{Helstrom1969} provides a lower bound on the symmetric error probability (which is in essence the average error probability, if an equal prior is assumed):
\begin{align}
    \alpha(M_m) + \beta(M_m) \geq     \frac{1}{2} \left(1 - \frac{1}{2} \oneNorm{\rho_0^{\otimes m} - \rho_1^{\otimes m}}\right). \label{eq:Helstrom}
\end{align}


The Quantum Stein's lemma~\cite{hiai1991proper,ogawa2000strong}, on the other hand, provides an asymptotic performance bound for the asymmetric setting when a large number of copies are available. Specifically the lemma states that for all $\eps \in (0,1)$,
\begin{equation}
    \lim_{m\rightarrow \infty} \frac1m \log \beta^*(\epsilon) = - D(\rho_0\|\rho_1),
\end{equation}
where
\begin{equation}
    \beta^*(\epsilon) = \min_{M_m: \alpha(M_m)< \epsilon} \beta(M_m).
\end{equation}

The Quantum Stein's lemma provides an operational meaning for the quantum relative entropy that is a precise analog of the classical case. 

\subsection{Quantum Universal Hypothesis Testing}
\label{subsection:QUHT}

\ifarxiv
Before describing the QUHT problem, we find it helpful to revisit its  classical equivalent. 
Given samples $X^m \overset{\text{i.i.d.}}{\sim} Q$ from an unknown distribution, the goal is to determine if the true distribution is a known nominal distribution $P$ or not, i.e.
\begin{equation}
    \begin{split}
        H_0: Q = P\\
        H_1: Q \neq P.
    \end{split}
\end{equation}

This problem was considered in a seminal work by Hoeffding \cite{hoeffding1965asymptotically}, where he showed that the classical relative entropy, $D_{KL}(P\|Q)$, was, surprisingly, also the optimal achievable type II error exponent, when the type I error decayed in a bounded manner. One subtle note about his work was his consideration of the error under the alternative hypothesis, as he treated the true distribution as unknown but fixed. 
\fi

QUHT is a natural extension of the classical UHT framework to the quantum setting. 
Given an $m$-fold tensor product state $\sigma^{\otimes m}$, the goal is to determine if $\sigma$ is a known nominal state, $\rho$, i.e.
\begin{align}\label{equation:UHT}
    \begin{split}
    H_0: \sigma &= \rho\\
    H_1: \sigma &\neq \rho.
    \end{split}
\end{align}

The QUHT formulation directly mirrors Hoeffding’s formulation of the classical UHT, where the alternative is represented by an unknown but {\em fixed} state $\sigma$ that is different from $\rho$. 
For a given decision rule $M_m$, we denote the type I and II error probabilities as $\alpha(M_m)$ and $\beta(M_m)$, respectively. We note that while any decision rule should not depend on $\sigma$, as it is unknown apriori, the actual performance (in particular, type II error probability) is typically a function of $\sigma$. 
See \Cref{fig:oneSample} for a visual representation. 

The merit of the QUHT formulation lies in the generality of the alternative hypothesis. While this may motivate its use in broad applications of quantum systems, it also presents unique challenges. Many techniques developed for QHT rely on knowledge of both states — e.g., optimal tests for  structured alternatives \cite{hayashi2002optimal, bjelakovic2005quantum, notzel2014hypothesis} - and are thus not directly applicable when the alternative is entirely unspecified. An exception is found in recent work \cite{hayashi2024another}, which constructs a "quantum empirical distribution" and proves large-deviation results reminiscent of those in the classical setting. We further discuss its connection to our work in Section~\ref{section:Discussion}.

\subsection{Quantum Two-Sample Test}
\label{subsection:twosample}

The QUHT described in the previous subsection is often referred to as the one-sample test in the classical setting. Extending this formulation to the two-sample case is straightforward. Suppose now we have two sequences of quantum systems: $\rho_1,\ldots,\rho_m$ are generated from an unknown density matrix $\rho$ while $\sigma_1,\ldots,\sigma_n$ are generated from another unknown density matrix $\sigma$. The two-sample test is used to discriminate between the following two hypotheses:
\begin{align}\label{equation:two-sample}
    \begin{split}
    H_0: \sigma &= \rho\\
    H_1: \sigma &\neq \rho.
    \end{split}
\end{align}
While the form appears identical to the one-sample test (i.e., QUHT), the setup differs in one key way. The two-sample test has no knowledge of both $\rho$ and $\sigma$ and is given two sequences of samples, whereas the one-sample test has perfect knowledge of $\rho$ and uses a single sequence generated from the unknown state $\sigma$. 
See \Cref{fig:twoSampleVisual} for a visual of the problem and our proposed solution. 

\begin{subfigures}

\begin{figure}[hbt]
    \centering
    \ifarxiv
        \includegraphics[width=0.5\linewidth]{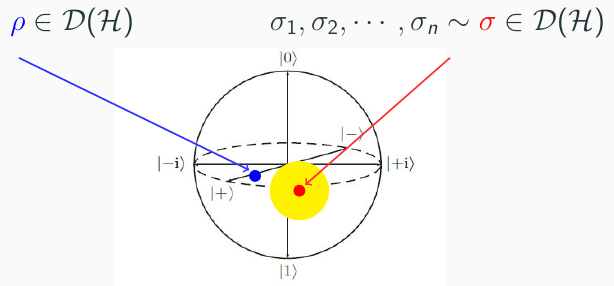}
    \fi
    \caption{\textbf{QUHT One Sample Problem}. 
    The yellow circle denotes the confidence region around our estimate $\hat{\sigma}$, with this region shrinking exponentially as $n$ increases. When $\rho$ is within this confidence region, we accept $H_0$. 
    }
    \label{fig:oneSample}
\end{figure}

\begin{figure}[hbt]
    \centering
    \ifarxiv
        \includegraphics[width=0.5\linewidth]{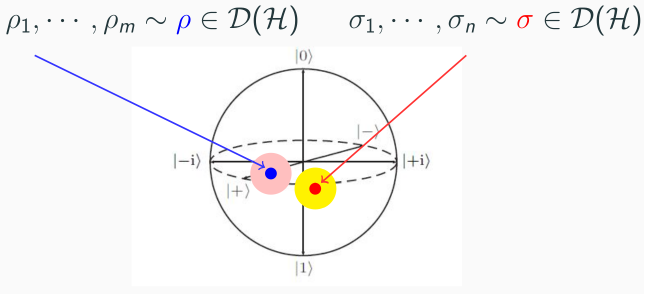}
    \fi
    \caption{\textbf{QUHT Two Sample Problem}. 
    Colored circles represent the estimated states of $\rho$ and $\sigma$, while the enclosing balls denote the confidence regions surrounding the estimates. Depending on the overlap between the confidence regions we may accept or reject $H_0$.
    }
    \label{fig:twoSampleVisual}
\end{figure}

\end{subfigures}


\section{Results}
\label{section:Results}


This section presents our main results on QUHT for both the one- and two-sample settings. As alluded to in Section~\ref{section:intro}, we primarily leverage the concentration results found in the QST literature to obtain the error exponent for type II error probabilities. The only exception is when the nominal state is pure, in which case we obtain an optimal error exponent. 
\ifarxiv
Proofs of all results can be found in the appendices of this paper. 

\fi


~\Cref{table:Summary_QUHT_Results} summarizes our main results for both one‑sample and two‑sample variants of QUHT.  To emphasize the dominant error‑exponent terms, we have omitted those that vanish asymptotically as the number of copies of the quantum states under test grows without bound.  Here, \(m\) and \(n\) denote the numbers of identical copies of the states \(\rho\) and \(\sigma\), respectively, as introduced in Section~\ref{subsection:twosample}.  Throughout our analysis, the type I error threshold \(\alpha\in(0,1)\) is treated as a fixed constant.  In the two‑sample setting, we restrict attention to the regime in which \(m\) and \(n\) increase at the same rate and define \(k=\min\{m,n\}\) to express all bounds in terms of the effective sample size. 

\begin{table}[hbt]
\caption{Table of QUHT Results. (1S: one-sample, 2S: two-sample)}
\label{table:Summary_QUHT_Results}
\ifarxiv
\centering
\resizebox{0.8\columnwidth}{!}{%
\begin{tabular}{|c|c|c|c|}
\hline
\textbf{Setting} & \textbf{Measurement Type} & \textbf{Tomography Source} & \textbf{Type II Error} \\ \hline
\begin{tabular}[c]{@{}c@{}}1S Qudit Pure State\\ \Cref{theorem:PureState_Symmetric}\end{tabular} & Indep. &  & $[F(\rho, \sigma)]^m$ \\ \hline
\begin{tabular}[c]{@{}c@{}}1S Qubit General State\\ \Cref{theorem:Qubit_PauliMeas}\end{tabular} & Indep. & \Cref{lemma:PauliConcentration} & $\Exp{-\frac{m \oneNorm{\rho - \sigma}^2}{54}}$ \\ \hline
\begin{tabular}[c]{@{}c@{}}2S Qubit General State\\ \Cref{theorem:Qubit_PauliMeas_TwoSample}\end{tabular} & Indep. & \Cref{lemma:PauliConcentration} & $\Exp{-\frac{k \oneNorm{\rho - \sigma}^2}{216}}$ \\ \hline
\begin{tabular}[c]{@{}c@{}}1S Qudit General State\\ \Cref{theorem:QUHT_OneSample_IndepMeas}\end{tabular} & Indep. & \Cref{lemma:Tomography_IndepMeas} & $\Exp{-\frac{m \oneNorm{\rho - \sigma}^2}{86 \cdot d^3}}$ \\ \hline
\begin{tabular}[c]{@{}c@{}}2S Qudit General State\\ \Cref{theorem:QUHT_TwoSample_IndepMeas}\end{tabular} & Indep. & \Cref{lemma:Tomography_IndepMeas} & $\Exp{-\frac{k \oneNorm{\rho - \sigma}^2}{344 \cdot d^3}}$ \\ \hline
\begin{tabular}[c]{@{}c@{}}1S Qudit General State\\ \Cref{theorem:QUHT_OneSample_EntMeas}\end{tabular} & Entangled & \Cref{lemma:Tomography_EntMeas} & $\Exp{-\frac{m\oneNorm{\rho - \sigma}^2}{2}}$ \\ \hline
\begin{tabular}[c]{@{}c@{}}2S Qudit General State\\ \Cref{theorem:QUHT_TwoSample_EntMeas}\end{tabular} & Entangled & \Cref{lemma:Tomography_EntMeas} & $\Exp{-\frac{k \oneNorm{\rho - \sigma}^2}{8}}$ \\ \hline
\end{tabular}%
}
\fi
\end{table}


\subsection{Pure State QUHT} 

We begin by establishing that for symmetric QUHT, we can achieve asymptotically optimal error performance by performing measurements in the eigenbasis of $\rho$ — a result that closely parallels known optimality results for pure state symmetric QHT \cite{kimmel2017hamiltonian} \cite{cheng2024invitation}.

\begin{theorem}[Symmetric Pure State]
    \label{theorem:PureState_Symmetric}


    Take $\rho = \ket{\phi}\bra{\phi} \in \mathcal{D}(\mathcal{H})$ a pure state density operator and $\sigma \in \mathcal{D}(\mathcal{H})$ a (potentially mixed) density operator. In the QUHT setting, given $\sigma^{\otimes m}$, there exists a decision rule $M_m$ with independent measurements, such that 
    { \small
    \begin{align}\label{equation:pure_state_error}
        \alpha(M_m) &= 0\\
        \beta(M_m) &= \left[F(\rho, \sigma)\right]^m.
    \end{align}
    }
\end{theorem}

For such a test, we have 
{\small
\begin{equation}
    \lim_{m\to \infty} \frac{1}{m} \ln\left(\alpha(M_m) + \beta(M_m)\right) =  -\ln F(\rho, \sigma).
\end{equation}}



\ifarxiv
In \Cref{appendix:PureStateSymmetric_Proof} we prove \Cref{theorem:PureState_Symmetric} and in \Cref{appendix:PureStateSymmetric_ErrorBounds}
we replicate 
\fi
a result from \cite{cheng2024invitation} which demonstrates our method is optimal in the sense that for any other decision rule $\tilde{M}_m$ we will have 
{ \small
\begin{equation}\label{equation:bounded_error}
    \liminf\limits_{m \to \infty} \frac{1}{m} \ln\left(\alpha(\tilde{M}_m) + \beta(\tilde{M}_m)\right)  \geq -\ln F(\rho, \sigma).
\end{equation}
}

Note that $\tilde{D}_{1/2}(\rho\| \sigma) := -\ln(F(\rho, \sigma))$ is the sandwiched Rényi Relative Entropy \cite{muller2013quantum}, \cite{wilde2014strong} between $\rho$ and $\sigma$ of order $1/2$.

\subsection{Qubit QUHT}
\label{subsection:Qubit_QUHT}

\ifarxiv
Next, we provide an example illustrating how leveraging a concentration result for a quantum state estimator allows us to derive a corresponding result in the QUHT framework. 
Let $I_2$ be the two dimensional identity matrix, and $\sigma_X, \sigma_Y, \sigma_Z$ be the standard Pauli matrices \cite{nielsen2010quantum}, and matrix $\vec{\sigma}=[\sigma_X,\sigma_Y,\sigma_Z]^T$. Any Hermitian $A\in L(\mathcal H_2)$ admits the Bloch decomposition 
\begin{equation}
    A=\frac{1}{2}\bigl(I_2 + \vec{r}\!\cdot\!\vec{\sigma}\bigr),
\end{equation}
where $\vec{r}=[\Tr[A\sigma_X],\Tr[A\sigma_Y],\Tr[A\sigma_Z]]$.

Let \(\sigma\in\mathcal{D}(\mathcal{H}_2)\) and \(m\) be divisible by three. Denote the joint state of \(m\) copies by \(\sigma^{\otimes m}\), and partition its factors into three sets \(G_i\) (\(i\in\{X,Y,Z\}\)) of size \(m/3\). On each \(G_i\), perform the \(0,1\)-valued projective measurement \(M_i\) in the eigen-basis of the Pauli operator \(\sigma_i\in\{\sigma_X,\sigma_Y,\sigma_Z\}\). The empirical Bloch coordinate is  
\( \hat r_i = \frac{3}{m}\sum_{\ell\in G_i}M_i(\sigma) \),
which satisfies \(\mathbb{E}[\hat r_i]=\Tr[\sigma_i\sigma]\), and define the vector $\hat{\vec{r}} = [\hat{r}_x, \hat{r}_y, \hat{r}_z]$. The reconstructed state is 
{\small
\begin{equation}
\hat\sigma = \tfrac12\Bigl(I_2 + \hat{\vec{r}}\!\cdot\!\vec{\sigma}\Bigr).
\end{equation}
}

Then we have the following concentration result:

\fi




\begin{lemma}[Pauli Operator Concentration]
    \label{lemma:PauliConcentration}

    For $\sigma$, $\hat{\sigma}$ as described above, and for $\eps > 0$, we have 
    
    {\small
    \begin{equation}\label{equation:Concentration}
        \Proba{\oneNorm{\sigma - \hat{\sigma}} > \eps} \leq 6\Exp{-\frac{m\eps^2}{54}}
    \end{equation}}
    
\end{lemma}

As the number of samples increases, the confidence region around $\sigma$ becomes more concentrated, with $\rho$ eventually residing outside it.

%
    \ifarxiv
    The proof of this result can be found in \Cref{appendix:Proof_PauliConcentration}.

    \fi

    

Using this result, we can recover a bound on universal hypothesis testing for (mixed) qubit systems. 

\begin{theorem}[Qubit One-Sample]
    \label{theorem:Qubit_PauliMeas}
    In the QUHT one-sample setting, for $\rho, \sigma \in \mathcal{D}(\mathcal{H}_2)$, given $\sigma^{\otimes m}$ and $\alpha \in (0,1)$, using only independent measurements there exists a decision rule $M_m$ so that for $m$ large enough, we will have 
    {
    \small
    \begin{align}\label{equation:TypeITypeII}
        \alpha(M_m) &\leq \alpha\\
        \beta(M_m) &\lesssim  \Exp{\frac{-m \oneNorm{\rho - \sigma}^2}{54}},
    \end{align}
    }
    where the $\lesssim$ hides an additive $O\left(m^{1/2}\right)$ term in the exponent.
\end{theorem}

We can use this concentration result to define a strategy for (mixed) two-sample qubit QUHT. 

\begin{theorem}[Qubit Two-Sample]
    \label{theorem:Qubit_PauliMeas_TwoSample}
    In the QUHT two-sample setting for $\rho, \sigma \in \mathcal{D}(\mathcal{H}_2)$, suppose $\sigma^{\otimes m}$, $\rho^{\otimes n}$, $\alpha \in (0,1)$ were given, and take $k = \min\{m, n\}$. For large enough $m, n$, and using independent measurements, there exists a decision rule $M_m$ such that 
    { 
    \small
    \begin{align}
        \alpha(M_k) &\leq \alpha\\
        \beta(M_k) &\lesssim \Exp{-\frac{k \oneNorm{\rho - \sigma}^2}{216}},
    \end{align}
    }
    where $\lesssim$ hides an additive $O(k^{1/2})$ term in the exponent. 
    
\end{theorem}

\subsection{Qudit Independent Measurements}
\label{subsection:QUHT_IndepMeas}

Unfortunately, extending the above result to higher-dimensional Hilbert spaces gives an error exponent that scales poorly with the dimension of the space. For example, any Hermitian operator on $b$ qubits ($d = 2^b$)
can be decomposed using $d^2 - 1$ components, from Pauli strings \cite{nielsen2010quantum}. Then applying the same method would give 
{\small
\begin{align}\label{equation:Construction_sigma}
    \Proba{\oneNorm{\sigma - \sh} \geq \eps} \leq 2(d^2 - 1) \Exp{-\frac{m\eps^2}{2(d^{2} - 1)^3}}.
\end{align}}

    
    

\ifarxiv
However, for general $d$-dimensional qudit states, we can do better using independent measurements by applying the following result from QST \cite{guctua2020fast}:
\begin{lemma} [Theorem 1 ~\cite{guctua2020fast}]
    
    \label{lemma:Tomography_IndepMeas}
    
    For a density operator $\tau \in \mathcal{D}(\mathcal{H})$, with $\dim(\mathcal{H}) = d$, and $\rank(\tau) = r$, it is possible to construct, via $m$ independent measurements, $\hat{\tau} \in \mathcal{D}(\mathcal{H})$ such that, for any $\delta \in (0,1)$,
    {\small
    \begin{equation}\label{equation:Tomography_IndepMeas_Weakened}
        \Proba{\oneNorm{\tau - \hat{\tau}} \geq \delta} \leq d \Exp{-\frac{m\delta^2}{86 r^2 d}}.
    \end{equation}
    }
\end{lemma}

Note that this is a special case of Guţă's \cite{guctua2020fast} more general result, where they take measurements that are $2$-designs \cite[Definition 1]{guctua2020fast}, and bounding the operators rank by $d$. This yields a concentration result that we can use in the same manner as before.

\fi



\begin{theorem}[Independent One-Sample]
    \label{theorem:QUHT_OneSample_IndepMeas}
    In the QUHT one-sample setting with $\rho, \sigma \in \mathcal{D}(\mathcal{H})$, given $\sigma^{\otimes m}$ and $\alpha \in (0,1)$, using only independent measurements there exists a decision rule $M_m$, so that for $m$ large enough, we will have  
    { 
    \small
    \begin{align}
        \alpha(M_m) &\leq \alpha\\
        \beta(M_m) &\lesssim \Exp{-\frac{m \oneNorm{\rho - \sigma}^2}{86d^3}},
    \end{align}
    }

    where the $\lesssim$ hides an additive $O\left(m^{1/2}\right)$ term in the exponent. 
\end{theorem}

This also yields a result for the two-sample QUHT setting.

\begin{theorem}[Independent Two-Sample]
    \label{theorem:QUHT_TwoSample_IndepMeas}
    In the QUHT two-sample setting with $\rho, \sigma \in \mathcal{D}(\mathcal{H})$, suppose $\sigma^{\otimes m}$, $\rho^{\otimes n}$, $\alpha \in (0,1)$ were given, and take $k = \min\{m, n\}$. For large enough $m, n$, and using independent measurements, there exists a decision rule $M_m$ such that 
    {\small
    \begin{align}
        \alpha(M_k) &\leq \alpha\\
        \beta(M_k) &\lesssim \Exp{-\frac{k \oneNorm{\rho - \sigma}^2}{344d^3}},
    \end{align}
    }
    where $\lesssim$ is hiding an additive $O(k^{1/2})$ term in the exponent. 
\end{theorem}

\subsection{Entangled Measurements}


 While independent measurements are easier to implement, it has been shown that entangled measurements achieve a sample complexity that is strictly better than that achievable by independent measurements \cite{haah2016sample}. 

\begin{lemma}[Section 2.3~\cite{haah2016sample}]
    \label{lemma:Tomography_EntMeas}
    For a density operator $\tau \in \mathcal{D}(\mathcal{H})$, with $\dim(\mathcal{H}) = d$, and $\rank(\tau) = r$, it is possible to construct a density operator $\hat{\tau} \in \mathcal{D}(\mathcal{H})$ using $\tau^{\otimes m}$, such that for all $\delta \in (0,1)$, 
    \begin{equation}\label{equation:fidelity_concentration}
        \Proba{F(\tau, \hat{\tau}) \leq 1 - \delta} \leq (m+1)^{3rd} \Exp{-2m\delta}.
    \end{equation}
\end{lemma}
We further weaken this bound to the following:
\begin{equation}
    \label{equation:Tomography_EntMeas_Weakened}
    \Proba{F(\tau, \hat{\tau}) \leq 1 - \delta} \leq (m+1)^{3d^2}\Exp{-2m\delta}.
\end{equation}

This concentration result allows us to establish tighter bounds on the type II error rate. Specifically, we have the following one-sample result:

\begin{theorem}[Entangled One-Sample]
    \label{theorem:QUHT_OneSample_EntMeas}

    In the QUHT one-sample setting with $\rho, \sigma \in \mathcal{D}(\mathcal{H})$, given $\sigma^{\otimes m}$, $\alpha \in (0,1)$, and allowing measurements involving entanglement, there exists a decision rule $M_m$ such that for $m$ large enough, 
    {\small
    \begin{align}\label{equation:onesample_setting}
        \alpha(M_m) &\leq \alpha\\
        \beta(M_m) &\lesssim \Exp{-\frac{m\oneNorm{\rho - \sigma}^2}{2}},
    \end{align}
    }

    where the $\lesssim$ hides an additive $O\left(m^{1/2}\right)$ term in the exponent. 
\end{theorem}

We can also use this to get results for two-sample QUHT. 
In particular, using entangled measurements, we can achieve the following:

\begin{theorem}[Entangled Two-Sample]
    \label{theorem:QUHT_TwoSample_EntMeas}

    In the QUHT two-sample setting with $\rho, \sigma \in \mathcal{D}(\mathcal{H})$, suppose we are given $\sigma^{\otimes m}$, $\rho^{\otimes n}$, $\alpha \in (0,1)$, and take $k = \min\{m, n\}$. For large enough $m, n$, and using entangled measurements, there exists a decision rule $M_m$ such that 
    {
    \begin{align}\label{equation:entangled_setting}
        \alpha(M_m) &\leq \alpha\\
        \beta(M_m) &\lesssim \Exp{-\frac{k \oneNorm{\rho - \sigma}^2}{8}},
    \end{align}
    }
    where $\lesssim$ hides an additive $O(k^{1/2})$ term in the exponent. 
\end{theorem}

\section{Discussion}
\label{section:Discussion}

Recall that in the classical universal setting, Hoeffding was able to characterize the error exponent in terms of the KLD between the nominal and true distributions. In this work, we have established positive error exponents not in terms of the KLD, but the trace distance between the nominal and true distributions. Specifically, we showed various results of the following form: 

\begin{equation}
    \liminf_{n \to \infty}- \frac{1}{m} \ln(\beta(M_m))  \geq \frac{\oneNorm{\rho - \sigma}^2}{c},
\end{equation}
where $c \in \R$ depends on the dimension of the underlying Hilbert space, and the types of measurements used. This general form applies to both the one-sample and two-sample settings, with the latter offering greater flexibility in scenarios where the nominal state is itself inaccessible or generated from experimental data.


Comparing with the classical results for the universal hypothesis testing problem, which achieves an error exponent in the form of the relative entropy between $\rho$ and $\sigma$, this result appears somewhat lacking. For example, we have quantum Pinsker's inequality \cite[Theorem 3.1]{hiai1981sufficiency}, which establishes that the squared trace distance forms a lower bound for the relative entropy. 
Since Hoeffding establishes his result in the classical setting using a large deviation type bound, a first attempt at a quantum generalization might utilize quantum variations of the classical Sanov's theorem \cite[Theorem 11.4.1]{EInfo06}. Unfortunately, the results in \cite{bjelakovic2005quantum, notzel2014hypothesis} are insufficient for recovering a bound in the QUHT setting, as they require knowledge of the alternative hypothesis (the true state $\sigma$ in this work). 

As a corollary to \cite[Lemma 1, Corollary 6]{hayashi2024another}, one could get a decision rule $\hat{M}_{m}$ such that 
{ 
\begin{align}
    \alpha(\hat{M}_m) &\leq \alpha\\
    \liminf_{m \to \infty} -\frac{1}{m} \ln(\beta(\hat{M}_m)) &\geq - \hat{D}(\rho, \sigma),
\end{align}
}

where $\hat{D}(\rho\| \sigma)$ is defined as
{ 
\begin{equation}
    \hat{D}(\rho \| \sigma) = \lim_{s \to 0} \left\{\frac{-\ln\left(\Tr\left[ \left(\sigma^{\frac{1-s}{2s}} \rho \sigma^{\frac{1-s}{2s}}\right)^s \right]\right)}{s}\right\}.
\end{equation}
}

\ifarxiv
We briefly show the result in \Cref{appendix:AlternativeUsingHiyashiResult}. 
\fi
However, \cite[Lemma 3]{hayashi2024another} showed that $\hat{D}(\rho \| \sigma) \leq D(\rho\| \sigma)$, and tight relationships between $\hat{D}(\rho \| \sigma)$ and the trace norm are not  known, which makes direct comparison difficult. Nevertheless, the quantum empirical distribution they define is likely worth considering for establishing tighter bounds for QUHT.

If we could find a decision rule $M_m^*$ with
\begin{equation}
    \lim_{m \to \infty} \left\{\frac{1}{m} \ln(\beta(M_m^*))\right\} = -D(\rho\| \sigma),
\end{equation}
then applying 
Quantum Pinsker's Inequality, which is tight up to the constant \footnote{When $\rho, \sigma$ commute, the inequality reduces to the classical version, which is known to be tight to the constant \cite{csiszar2011information}},
we would have $\beta(M^*_m) \simeq e^{-mD(\rho, \sigma)} \leq e^{-\frac{m \oneNorm{\rho - \sigma}^2}{2}}$. However, in \Cref{theorem:QUHT_OneSample_EntMeas}, we showed that there exists a decision rule $M_m$ such that $\beta(M_m) \simeq e^{-\frac{m \oneNorm{\rho - \sigma}^2}{2}}$. In other words, even if a decision rule achieving an error exponent of $D(\rho \| \sigma)$ were available, its implied performance in terms of trace distance would be upper-bounded by a term matching the exponent achieved using tomography. 
\ifarxiv
We corroborate this in \Cref{appendix:optimalConstant} where we show that any result achieving a constant type I error rate, and a type II error exponent of $C\oneNorm{\rho - \sigma}^2$ with $C > \frac{1}{2}$ would violate the converse result of \textit{classical Stein's lemma}. 

\fi

Several other works have considered similar problems under different names in other areas of quantum information science \cite{long2024equivalence} \cite{burgholzer2020advanced} \cite{huang2024certifying} \cite{yu2022statistical}, often motivated by practical needs for verifying quantum devices or certifying quantum behavior in minimally structured settings. By framing QUHT as a formal hypothesis testing problem, we hope that it will provide both conceptual clarity and a foundation for future work. While our focus has been on leveraging QST to achieve provable error exponents, the broader framework of one- and two-sample QUHT may have applications analogous to those of classical universal hypothesis testing.




\bibliographystyle{IEEEtran}
\bibliography{IEEEabrv,refs.bib}

\ifarxiv
\appendix

\section{Proofs}
\label{section:Proofs}

We will make heavy use of the following relationship between the trace distance and the fidelity of a quantum state \cite{wilde2013quantum}
\begin{lemma}
    \label{lemma:FidelityTraceNormRelation}
    $\rho, \sigma \in \mathcal{D}(\mathcal{H})$, then 
    \begin{equation}
        \frac{1}{2} \oneNorm{\rho - \sigma} \leq \sqrt{1 - F(\rho, \sigma)}.
    \end{equation}
\end{lemma}

We also have the following quantum version of Pinsker's inequality \cite[Theorem 3.1]{hiai1981sufficiency} 
\begin{lemma}
    \label{lemma:QuantumPinskersInequality}
    \begin{equation}
        D(\rho\| \sigma) \geq \frac{1}{2} \oneNorm{\rho - \sigma}^2.
    \end{equation}
\end{lemma}


\subsection{One-Sample Generic Result}
\label{appendix:OneSampleGeneric}

Here we will prove a generic result which shows that any  concentration result for quantum state tomography of a particular form gives us a test achieving a certain type two error exponent in the  QUHT setting.

\begin{lemma}
    \label{lemma:OneSamp_GenericConcentrationResult}
    Given $\rho$ the density operator of $\sigma$ under $H_0$, suppose there is a method of constructing $\sh_m$ from $\sigma^{\otimes m}$ such that a concentration result of the following form holds: $\forall t > 0$, and for a constant $C$ and a function $g(m) = \poly(m)$
    \begin{equation}
        \label{equation:GenericConcentrationResult}
        \Proba{\oneNorm{\sh_m - \sigma} \geq t} \leq g(m) \Exp{-m C t^2}
    \end{equation}

    Then using the test 
    \begin{align}
        \hat{H}_0&: \oneNorm{\sh_m - \rho} \leq c_m\\
        \hat{H}_1&: \oneNorm{\sh_m - \rho} > c_m
    \end{align}

    Given $\alpha_m > 0$, if we took $c_m = O\left(\frac{ \ln^{1/2}(1/\alpha)}{\sqrt{m}}\right)$, we would get  
    \begin{equation}
        \Proba{\hat{H}_1 | H_0} \leq \alpha
    \end{equation}

    And for $m$ large enough, 
    \begin{equation}
        \Proba{\hat{H}_0 | H_1} \leq g(m) \Exp{-m\  C \oneNorm{\rho - \sigma}^2 + O\left(m^{1/2}\right)}
    \end{equation}
\end{lemma}

\begin{proof}
    The type I error can be written in terms of the probability of having $\sh_m$ far from $\sigma$, which we can bound directly using the concentration result. Specifically, we have 
    \begin{align}
        \Proba{\hat{H}_1 | H_0} &= \Prob\big[\oneNorm{\sh_m - \rho} > c_m\  \big| \ \sigma = \rho \big]\\
        &=\Proba{\oneNorm{\sh_m - \sigma} > c_m}\\
        &\leq g(m) \Exp{-m \ C\  c_m^2}
    \end{align}

    Then taking 
    \begin{align}
        c_m &= \left(\frac{\ln(g(m)/\alpha)}{C m}\right)^{1/2}\\
        &= O\left(\frac{\ln(g(m)) \cdot \ln^{1/2}(1/\alpha)}{m^{1/2}}\right)
    \end{align}

    We would get $\Proba{\hat{H}_1 | H_0} \leq \alpha$. 

    For the type II error, we have 
    \begin{align}
        \Proba{\hat{H}_0 | H_1} &= \Proba{\oneNorm{\sh_m - \rho} \leq c_m}\\
        &\leq \Proba{ \oneNorm{\sigma - \rho} - \oneNorm{\sh_m - \sigma} \leq c_m }\\
        &= \Proba{ \oneNorm{\sh_m - \sigma} \geq \oneNorm{\sigma - \rho} - c_m }
    \end{align}

    And since $\rho \neq \sigma$ and $\sigma, \rho$ fixed in this setting, and since $c_m \to 0$ as $m \to \infty$, then for $m$ large enough $c_m < \oneNorm{\rho - \sigma}$. When this condition is satisfied, we have 
    \begin{align}
        \Proba{\hat{H}_0 | H_1} \leq g(m) \Exp{-C m \left(\oneNorm{\rho - \sigma} - c_m\right)^2}\\
        \leq g(m) \Exp{-C \ m \left(\oneNorm{\rho - \sigma}^2 - 2c_m\right)}\\
        = \Exp{-m \ C \oneNorm{\rho - \sigma}^2 + \frac{2 m^{1/2}  \ln^{1/2}(1/\alpha)}{C^{1/2} } + \ln(g(m)) }
    \end{align}
\end{proof}

Note that \Cref{theorem:Qubit_PauliMeas} and   \Cref{theorem:QUHT_OneSample_IndepMeas} are direct corollaries of this result using \Cref{lemma:PauliConcentration} and \Cref{lemma:Tomography_IndepMeas} respectively. And as corollary of \Cref{lemma:Tomography_EntMeas}, we have the following concentration result in terms of the trace norm:
\begin{corollary}
    \label{corollary:Tomography_EntMeas_TraceNorm}
    Using \Cref{lemma:Tomography_EntMeas}, we can construct $\sh_m$ using $\sigma^{\otimes m}$ such that 
    \begin{equation}
        \Proba{\oneNorm{\sigma - \sh_m} \geq t} \leq (m+1)^{3d^2} \Exp{-\frac{m t^2}{2}}
    \end{equation}
\end{corollary}

\begin{proof}
    \begin{align}
    \Proba{F(\sigma, \sh_m) \leq 1 - \delta} &\geq \Proba{1 - \frac{1}{4}\oneNorm{\sigma - \sh_m}^2 \leq 1 - \delta}\\
    &= \Proba{\oneNorm{\sigma - \sh_m} \geq 2 \delta^{1/2}}\\
    \implies \Proba{\oneNorm{\sigma - \sh_m} \geq t} &\leq (m+1)^{3d^2} \Exp{-\frac{m t^2}{2}}\\
\end{align}
\end{proof}

And \Cref{corollary:Tomography_EntMeas_TraceNorm} combined with \Cref{lemma:OneSamp_GenericConcentrationResult} gives us \Cref{theorem:QUHT_OneSample_EntMeas}. 

\subsection{Two-Sample Generic Result}
\label{appendix:TwoSampleGeneric}

And in a similar manner, we can prove that the same type of generic concentration result yields a test with a generic bound on the type II error. 

\begin{lemma}
    \label{lemma:TwoSamp_GenericConcentrationResult}
    Given the same general concentration result as in   \Cref{equation:GenericConcentrationResult} for the two sample setting with $k = \min\{m,n\}$, we could use the decision region
    \begin{align}
        \hat{H}_0 &: \oneNorm{\sh_m - \rh_n} \leq c_k \\
        \hat{H}_1 &: \oneNorm{\sh_m - \rh_n} > c_k
    \end{align}

    We assume that $\limsup \frac{n}{m} = a < +\infty$, i.e. that $n$ and $m$ go to $\infty$ with bounded ratio. 
    
    And given $\alpha > 0$, taking $c_k = O\left(\frac{\ln((g(n) + g(m))/\alpha)}{k^{1/2}}\right)$ would yield 
    \begin{equation}
        \Proba{\hat{H}_1 | H_0} \leq \alpha
    \end{equation}

    And 
    {\small
    \begin{equation}
        \Proba{\hat{H}_0 | H_1}  \leq (g(n) + g(m)) \Exp{-k \frac{C \oneNorm{\rho - \sigma}^2}{4} + O\left(k^{1/2}\right)}
    \end{equation}
    }
\end{lemma}

\begin{proof}
    \begin{align}
        &\Proba{\hat{H}_1| H_0} \\
        &= \Proba{\oneNorm{\sh_m - \sh_n} > c_k}\\
        &\leq \Proba{\oneNorm{\sh_m - \sigma} + \oneNorm{\sh_n - \sigma} > c_k}\\
        &\leq \Proba{\oneNorm{\sh_m - \sigma} > \frac{c_k}{2}} + \Proba{\oneNorm{\sh_n - \sigma} > \frac{c_k}{2}}\\
        &\leq g(m) \Exp{-m\ C \frac{c_k^2}{4}} + g(n) \Exp{-m C\frac{c_k^2}{4}}\\
        &\leq \left(g(n) + g(m)\right)\Exp{-k C \frac{c_k^2}{4}}
    \end{align}

    where $k = \min(m, n)$. 

    Then if we have  
    \begin{align}
        c_k &= \sqrt{\frac{4\ln\left(\frac{g(n) + g(m)}{\alpha}\right)}{Ck}}\\
        &= O\left(\frac{\ln^{1/2}(1/\alpha)}{k^{1/2}}\right)
    \end{align}

    Where we can ignore the $\ln(g(n) + g(m))$ terms so long as as $\limsup \frac{n}{m} = a < +\infty$, so that $\ln(g(n) + g(m))/k \to 0$. Then $c_k \to 0$ as $k \to \infty$

    Now, targeting the type II error, we assume that $\rho$ and $\sigma$ are fixed, and $\rho \neq \sigma$. 

    \begin{align}
        &\Proba{\hat{H}_0 | H_1} \\
        &= \Proba{\oneNorm{\rh_n - \sh_m} \leq c_k}\\
        &\leq \Proba{\oneNorm{\rho - \sigma} - \oneNorm{\rh_n - \rho} - \oneNorm{\sh_m - \sigma} \leq c_k}\\
        &= \Proba{\oneNorm{\rh_n - \rho} + \oneNorm{\sh_m - \sigma} \geq \oneNorm{\rho - \sigma} - c_k}\\
        &\leq \Proba{\oneNorm{\rh_n - \rho} \geq \frac{\oneNorm{\rho - \sigma} - c_k}{2} } \\
        &\ \ + \Proba{\oneNorm{\sh_m - \sigma} \geq \frac{\oneNorm{\rho - \sigma} - c_k}{2}}
    \end{align}

    And since $\rho \neq \sigma$, for $k$ large enough we would have $\oneNorm{\rho - \sigma} - c_k > 0$. Then using the concentration result, we get 
    \begin{align}
        &\Proba{\hat{H}_0 | H_1} \\
        &\leq g(n) \Exp{-n C \frac{(\oneNorm{\rho - \sigma} - c_k)^2}{4}} \\
        &\ \ \ + g(m) \Exp{-mC \frac{ (\oneNorm{\rho - \sigma} - c_k)^2 }{ 4 }}\\
        &\leq \left(g(n) + g(m)\right) \Exp{-kC \frac{\oneNorm{\rho - \sigma}^2  - 2 c_k}{4}}
    \end{align}
\end{proof}

And using \Cref{lemma:TwoSamp_GenericConcentrationResult}, we have \Cref{theorem:Qubit_PauliMeas_TwoSample} \Cref{theorem:QUHT_TwoSample_IndepMeas} and \Cref{theorem:QUHT_TwoSample_EntMeas} as corollaries of using \Cref{lemma:PauliConcentration} \Cref{lemma:Tomography_IndepMeas} and \Cref{lemma:Tomography_EntMeas} respectively.

\subsection{Proof of \texorpdfstring{\Cref{theorem:PureState_Symmetric}}{Theorem 1 (Symmetric Pure State)}}
\label{appendix:PureStateSymmetric_Proof}

Since $\rho$ is a pure state, there exists a unitary operator $U \in \mathcal{L}(\mathcal{H})$ such that $U \rho = \ket{0}\bra{0}$. Similarly, we can apply $U^{\otimes m}$ to $\sigma^{\otimes m}$, such that without loss of generality we can consider $\rho$ to be the pure state $\ket{0}\bra{0}$ (similarly, in the one sample case we can always assume without loss of generality that $\rho$ is diagonalizable in the computational basis). 

Now for $\tau \in \mathcal{D}(\mathcal{H})$, define $M_{\ket{0}}(\tau)$ to be the result of measuring $\tau$ using the positive operator valued measure (POVM) $\{\ket{0}\bra{0}, I - \ket{0}\bra{0}\}$, such that $M_{\ket{0}}(\tau) = 1$ if the outcome is $\ket{0}$, and 0 otherwise. For $i \in [m]$, take $\sigma_i$ to be the $i^{\text{th}}$ copy of $\sigma$ in $\sigma^{\otimes m}$. Then our decision rule would be 
\begin{align}
    \begin{split}
    \hat{H}_0: &\sum_{i=1}^m M_{\ket{0}}(\sigma_i) = m\\
    \hat{H}_1: &\sum_{i=1}^{m} M_{\ket{0}}(\sigma_i) < m.
    \end{split}
\end{align}

We naturally have $\alpha(M_m) = 0$. And by the Born rule, we have 
\begin{align}
    \Proba{M_{\ket{0}}(\sigma_i) = 1} = F(\rho, \sigma),
\end{align}

so that for the type II error, by the independence of the measurements 
\begin{align}
    \beta(M_m) &= \Proba{\sum_{i=1}^m M_{\ket{0}} (\sigma_i) = m}\\
    &= \left[\prod_{i=1}^m \Proba{M_{\ket{0}} (\sigma_i)=1}\right]\\
    &= \left[F(\rho, \sigma)\right]^m.
\end{align}

\subsection{Tightness of Symmetric Hypothesis Testing Error Bound}
\label{appendix:PureStateSymmetric_ErrorBounds}

The following result is essentially what was shown in Appendix F of \cite{cheng2024invitation}.

From \cite{holevo1973statistical} \cite{Helstrom1969} we know that the sum of type I and II errors is lower-bounded by $\eps^*_m := \frac{1}{2} \left(1 - \frac{1}{2} \oneNorm{\rho^{\otimes m} - \sigma^{\otimes m}}\right)$. Then consider the following argument: 

\begin{align}
    \eps^*_m &= \frac{1}{2}\left(1 - \frac{1}{2}\oneNorm{\rho^{\otimes m} - \sigma^{\otimes m} }\right)\\
    &= \frac{1}{2} \cdot \frac{(1 - \frac{1}{2}\oneNorm{\rho^{\otimes m} - \sigma^{\otimes m}}) (1 + \frac{1}{2}\oneNorm{\rho^{\otimes m} - \sigma^{\otimes m}})}{1 + \frac{1}{2}\oneNorm{\rho^{\otimes m} - \sigma^{\otimes m}}}\\
    &= \frac{1}{2} \cdot \frac{1 - \frac{1}{4}\oneNorm{\rho^{\otimes m} - \sigma^{\otimes m}}^2}{1 + \frac{1}{2} \oneNorm{\rho^{\otimes m} - \sigma^{\otimes m}}}\\
    &\geq \frac{1}{2} \cdot\frac{1 - \frac{1}{4}\oneNorm{\rho^{\otimes m} - \sigma^{\otimes m}}^2}{2}\\
    &\geq \frac{1}{4} F(\rho^{\otimes m}, \sigma^{\otimes m})
    \label{equation:EquationToReference}\\
    &= \frac{1}{4} \left[F(\rho, \sigma)\right]^m,
\end{align}

where line \ref{equation:EquationToReference} comes from \Cref{lemma:FidelityTraceNormRelation}. 

\subsection{Proof of \texorpdfstring{\Cref{lemma:PauliConcentration}}{Lemma 1 (Pauli Operator Concentration)}}
\label{appendix:Proof_PauliConcentration}

\begin{proof}

    For all $K \in \{X,Y, Z\}$, suppose that 
    \begin{equation}
        \label{equation:InitialAssumptionPauliConcentrationProof}
        \Proba{\frac{1}{2} \oneNorm{\sigma_K - \hat{\sigma}_K} > \eps} \leq 2 \Exp{-\frac{m\eps^2}{6}};
    \end{equation}

    then we have 
    { \small
    \begin{align}
        \oneNorm{\sigma - \hat{\sigma}} &= \frac{1}{2}\oneNorm{(\sx - \shx) + (\sy - \shy) + (\sz - \shz)}\\
        &\leq \frac{1}{2}\oneNorm{\sx - \shx} + \frac{1}{2}\oneNorm{\sy - \shy} + \frac{1}{2}\oneNorm{\sz + \shz},
    \end{align}
    } 

    Therefore
    { 
    \begin{align}
        &\implies \Proba{\oneNorm{\sigma - \hat{\sigma}} \geq \eps} \\
        &\leq \Proba{\frac{1}{2} (\oneNorm{\sx - \shx} + \oneNorm{\sy - \shy} + \oneNorm{\sz - \shz}) \geq \eps}\\
        &\leq \Proba{\frac{1}{2} \oneNorm{\sx - \shx} \geq \frac{\eps}{3}} + \Proba{\frac{1}{2} \oneNorm{\sy - \shy} \geq \frac{\eps}{3}} \\
        &\ \ \ + \Proba{\frac{1}{2}\oneNorm{\sz - \shz} \geq \frac{\eps}{3}}\\
        &\overset{(a)}{\leq} 3 \left(2 \Exp{-\frac{m\eps^2}{54}}\right)\\
        &= 6 \Exp{-\frac{m\eps^2}{54}},
    \end{align}
    }

    where $(a)$ comes from the application of \Cref{equation:InitialAssumptionPauliConcentrationProof}. 

    Now, all that remains is to prove \Cref{equation:InitialAssumptionPauliConcentrationProof}. Note that 
    \begin{align}
        \frac{1}{2}\oneNorm{\sk - \shk} &= \frac{1}{2}\oneNorm{K \Tr[K\sigma] - K \hat{p}_K(\sigma)}\\
        &= \frac{1}{2}\abs{\Tr[K\sigma] - \hat{r}_K(\sigma)} \cdot \oneNorm{K}\\
        &= \abs{\Tr[K\sigma] - \hat{r}_K(\sigma)}\\
        &= \abs{\E[\hat{r}_K(\sigma)] - \hat{r}_K(\sigma)}.
    \end{align}

    So 
    \begin{align}
        \Proba{\frac{1}{2} \oneNorm{\sk - \shk} \geq \eps} &= \Proba{\abs{\E[p_K(\sigma)] - \hat{p}_K(\sigma) }}\\
        &\leq 2 \Exp{-\frac{m \eps^2}{6}}
    \end{align}
    where the last inequality is an application of Hoeffding's Inequality \cite{hoeffding1994probability}. 
\end{proof}

\subsection{Sanov's Theorem Result}
\label{appendix:AlternativeUsingHiyashiResult}

By \cite[Lemma 1]{hayashi2024another}, there exists a sequence of POVMs $\{A_{\rho, r_m}, I - A_{\rho, r_m}\}_m$ such that 
\begin{equation}
    \Tr[\rho^{\otimes m} (I - A_{\rho, r_m})] \leq (m+1)^{\frac{(d+4)(d-1)}{2}} e^{-m \cdot r_m}.
\end{equation}

Describing the POVM using notation in the paper, we would have
\begin{equation}
    A_{\rho, r_m} = \sum_{(p', \rho') \in S^{c}_{\rho, r_m} \cap \mathcal{R}_m[B]} T_{np', \rho', B}^n.
\end{equation}

Then we define $r_m$ so that the type I error is controlled, i.e. 
\begin{align}
    r_m &= \frac{(d+4)(d-1)}{2} \frac{\ln(m+1)}{m} + \frac{\ln(1/\alpha)}{m}\\
    &= O\left(\frac{\ln(m)}{m}\right),
\end{align}

so that $\Tr[\rho^{\otimes m} (I - A_{\rho, r_m})] \leq \alpha$. 

Now, if we assume that $\sigma \neq \rho$, then using \cite[Corollary 6]{hayashi2024another}, for $0 < R < \hat{D}(\rho\| \sigma)$, there exists $m$ large enough ($r_m$ small enough) such that 
\begin{align}
    \lim_{m \to \infty} \left\{\frac{1}{m} \ln\left(\Tr[\sigma^{\otimes m} A_{\rho, r_m}]\right)\right\} \leq -R.
\end{align}

Then, if we take our decision rule to be the measurement associated with this POVM, it would yield the stated result.

\subsection{Proof of Optimal Constant}
\label{appendix:optimalConstant}

    The (classical) hypothesis testing problem addressed here is 
    \begin{equation}
        \label{equation:ClassicalSimpleHTest}
        \begin{split}
            H_0: Q &= P_0\\
            H_1: Q &= P_1
        \end{split}
    \end{equation}
    
    \begin{claim}
        Let $\eps \in (0,1)$ be our type I error constraint.
        
        For any hypothesis test, $M_n$, of  \ref{equation:ClassicalSimpleHTest}, achieving
        \begin{align}
            \alpha(M_n) &\leq \eps\\
            \liminf_{n \to \infty} \frac{-1}{n} \ln \beta(M_n) &\geq C \oneNorm{P_0 - P_1}^2,
        \end{align}
        for some $C > 0$. Of course the type I error constraint is our
    
        The best possible $C$ that can be achieved across all possible $P_0, P_1 \in \Pcal^d$ is $C = \frac{1}{2}$.
    \end{claim}
    
    We need two things to prove the result. First, we need Pinsker's inequality along with the fact that it is sharp, in the following sense \cite{csiszar2011information}: For any $\eps > 0$, $\exists P,Q \in \Pcal^d$ such that 
    \begin{equation}
        \label{equation:PinskersSharp}
        \frac{1}{2} \leq \frac{D(P\|Q)}{\oneNorm{P -Q}^2} \leq \frac{1}{2} + \eps
    \end{equation}
    
    We will also need the converse portion of Stein's Lemma, which states that for all tests $M_n$ achieving a certain type I error constraint, they must have \cite{cover1999elements}
    \begin{equation}
        \label{equation:SteinsConverse}
        \limsup_{n \to \infty} \frac{-1}{n} \ln \beta(M_n) \leq D(P_0 \| P_1)
    \end{equation}

    \begin{proof}

        Suppose for contradiction that there exists a test $M_n$, such that $\forall P_0, P_1 \in \Pcal^d$, it achieves the type I error constraint, and for some $C > \frac{1}{2}$ we get 
        \begin{equation}
            \label{equation:falseResult}
            \liminf_{n \to \infty} \frac{-1}{n} \ln \beta(M_n) \geq C \oneNorm{P_0-P_1}
        \end{equation}
    
        Let $\delta = C - \frac{1}{2} > 0$. Take $0 < \eps < \delta$, and then using \ref{equation:PinskersSharp}, find $\dot{P}, \dot{Q} \in \Pcal^d$ such that 
        \begin{equation}
            \label{equation:UsingSharpPinskers}
            \frac{D(\dot{P}\| \dot{Q}) \ \ }{\oneNorm{\dot{P} - \dot{Q}}^2} \leq \frac{1}{2} + \eps
        \end{equation}

        Then applying \ref{equation:falseResult} for $\dot{P}, \dot{Q}$, we get:
        \begin{align}
            \limsup_{n \to \infty} \frac{-1}{n} \ln \beta(M_n) &\geq \liminf_{n \to \infty} \frac{-1}{n} \ln \beta(M_n)\\
            &\geq C \oneNorm{\dot{P} - \dot{Q}}^2\\
            &= \frac{1 + 2\delta}{2}\oneNorm{\dot{P} - \dot{Q}}^2 \\
            &\geq \left(\frac{1 + 2 \delta}{2}\right) \left(\frac{2}{1 + 2\eps}\right) D(\dot{P}\| \dot{Q})\\
            &> D(\dot{P} \| \dot{Q})
        \end{align}
    
        which violates \ref{equation:SteinsConverse}, thus giving us a contradiction. 
        
    \end{proof}

\fi

\end{document}